\documentclass[12pt]{extarticle}
\usepackage{extsizes}
\usepackage[reqno]{amsmath} 
\usepackage{amsmath}
\usepackage{amssymb}
\usepackage{amsthm}
\usepackage{amscd}
\usepackage{amsfonts}
\usepackage{graphicx}
\usepackage{dsfont}
\usepackage{fancyhdr}
\usepackage{enumerate}
\usepackage{youngtab}
\usepackage[utf8]{inputenc}
\usepackage{comment}

\usepackage{subfigure}
\usepackage[margin=1.5in]{geometry}
\usepackage{graphicx}
\usepackage{algorithm}
\usepackage{algpseudocode}
\usepackage{color}
\usepackage{notation}
\usepackage{hyperref}

\theoremstyle{theorem}
\newtheorem{corollary}{Corollary}

\newtheorem{lemma}[corollary]{Lemma}

\newtheorem{proposition}[corollary]{Proposition}
\newtheorem{theorem}[corollary]{Theorem}

\newcommand\seq[3]{\{#1_{#2},\ldots,#1_{#3}\}}

\begin{document}

\AtEndDocument{%
  \par
  \medskip
  \begin{tabular}{@{}l@{}}%
    \textsc{Gabriel Coutinho}\\
    \textsc{Dept. of Computer Science} \\ 
    \textsc{Universidade Federal de Minas Gerais, Brazil} \\
    \textit{E-mail address}: \texttt{gabriel@dcc.ufmg.br} \\ \ \\
    \textsc{Pedro Ferreira Baptista} \\
    \textsc{Dept. of Computer Science} \\ 
    \textsc{Universidade Federal de Minas Gerais, Brazil} \\
    \textit{E-mail address}: \texttt{pedro.baptista@dcc.ufmg.br}\\ \ \\
    \textsc{Chris Godsil} \\
    \textsc{Dept. of Combinatorics and Optimization} \\ 
    \textsc{University of Waterloo, Canada} \\
    \textit{E-mail address}: \texttt{cgodsil@uwaterloo.ca}\\ \ \\
    \textsc{Thomás Jung Spier} \\
    \textsc{Dept. of Computer Science} \\ 
    \textsc{Universidade Federal de Minas Gerais, Brazil} \\
    \textit{E-mail address}: \texttt{thomasjspier00@gmail.com}\\ \ \\
    \textsc{Reinhard Werner} \\
    \textsc{Institut für Theoretische Physik} \\ 
    \textsc{Leibniz Universität Hannover, Germany} \\
    \textit{E-mail address}: \texttt{reinhard.werner@itp.uni-hannover.de}
  \end{tabular}}

\title{Irrational quantum walks}
\author{Gabriel Coutinho\footnote{gabriel@dcc.ufmg.br --- remaining affiliations in the end of the manuscript.} \and Pedro Ferreira Baptista \and Chris Godsil \and Thomás Jung Spier \and Reinhard Werner}
\date{\today}
\maketitle
\vspace{-0.8cm}

\begin{abstract} 
	The adjacency matrix of a graph $G$ is the Hamiltonian for a continuous-time quantum walk
	on the vertices of $G$. Although the entries of the adjacency matrix are integers,
	its eigenvalues are generally irrational and, because of this, the behaviour of the walk
	is typically not periodic. In consequence we can usually only compute numerical
	approximations to parameters of the walk.
	In this paper, we develop theory to exactly study any quantum walk generated by an integral Hamiltonian. As a result, we provide exact methods to compute the average of the mixing matrices, and to decide whether pretty good (or almost) perfect state transfer occurs in a given graph. We also use our methods to study geometric properties of beautiful curves arising from entries of the quantum walk matrix, and discuss possible applications of these results.  
\end{abstract}

\begin{center}
average mixing matrix
\end{center}

\section{Introduction}

We will study continuous-time quantum walks on finite, undirected, unweighted graphs $G$. The Hamiltonian $H$ for such walks is defined by
\[
	H = \frac{1}{2} \sum_{ab \in E(G)} (X_aX_b + Y_aY_b).
\]
This Hamiltonian acts on $\Cds^{2^n}$, with $X_a$ and $Y_a$ denoting the operators that apply the Pauli matrices $X$ and $Y$ on the qubit located at vertex $a$ (and leave the other qubits invariant).

The $k$-excitation subspace is spanned by vectors $\ket{s}$ with $s \in \{0,1\}^n$ with $k$ entries equal to $1$. Here $\ket{s}$ is to be understood as the Kronecker product of $\ket 0$s and $\ket 1$s, where $\{\ket 0,\ket 1\}$ is the standard basis of $\Cds^2$. The Hamiltonian $H$ leaves each of the $k$-excitation subspaces invariant, and if $A=A(G)$, its action on the $1$-excitation subspace is determined by $A$. 
It follows (from Schr\"odinger's equation) that if the initial state of the walk in the $1$-excitation subspace is given by the unit vector $\ket u$, the state time $t$ is (essentially) $\exp(\ii t A)\ket u$.

Since the adjacency matrix is finite, real and symmetric it has real eigenvalues $\seq\theta0d$.
If $E_r$ represents orthogonal projection onto the $\theta_r$-eigenspace of $A$, then $A$ has the orthogonal projection
\[
	A = \sum_{r = 0}^d \theta_r E_r.
\]
Thus, the quantum walk is completely described by its transition matrix
\[
	\exp(\ii t A)  = \sum_{r = 0}^d \e^{\ii t \theta_r} E_r.
\]
For most graphs the eigenvalues $\theta_r$ are irrational, and so are the entries of the projectors $E_r$, thus forcing, in principle, that one deals with numerical approximations when concretely observing quantum walks. Yet, because the characteristic polynomial of $A$ is always monic and integral, it is possible to study the properties of the quantum walk over its splitting field. This is the main topic of this paper.

In Section \ref{sec:polynomials}, we recall basic facts about algebraic extensions of the rationals, followed by the presentation of a result  due to Landau \cite{Landau1985} which guarantees a relatively efficient complete factorization of a polynomial with integer coefficients over its splitting field. 
We use this to provide an algorithm that computes the entries of the average mixing matrix of the walk. It follows from a Galois theoretic argument that the entries of this matrix are rational, but no algorithm to compute its entries was known. (We defer the definition of the average mixing matrix, but note that it provides a guide to the long-term behaviour of continuous quantum walks.)
 
Section \ref{sec:kronecker} is devoted to the study of pretty good state transfer. This occurs when the probability of state transfer between two qubits in the continuous-time quantum walk network can be found to be arbitrarily close to 1 at different (and increasingly large) times; it is the epsilon variant of perfect state transfer. 
Perfect state transfer has been extensively studied, and there is a (polynomial time) algorithmic characterization of when it occurs (see \cite{CoutinhoGodsilPSTpolytime}). However, the best known tool to decide the existence of pretty good state transfer in a network relies on Kronecker's Theorem (see for instance the recent works \cite{VinetZhedanovAlmost,Banchi2017,GodsilKirklandSeveriniSmithPGST,CoutinhoGuovanBommel,vanbommel2020pretty,LippnerEisenbergKempton}), for which an algorithmic version has not yet been found, to the best of our knowledge. Landau's algorithm from Section \ref{sec:polynomials} allows us to use the Smith Normal Form to check Kronecker's condition, thus providing us an exact method that checks whether the condition in Kronecker's Theorem is satisfied, and thus whether pretty good state transfer occurs in a graph.
 
It $U(t)$ is the transition matrix $\exp(\ii t A)$ and $a,b\in V(G)$, then $U(t)_{a,b}$ describes a curve in the complex plane.
In Sections \ref{sec:shape}, \ref{sec:uninteresting} and \ref{sec:oddprime} we discuss the application of the results from the previous sections to describe geometric properties of these curves. We focus our examples and results on the cases where $G$ is an odd prime cycle, but the techniques generalize to any graph, and we see another example in Section \ref{sec:another}. In particular, we describe a method that determines the rotational symmetries and the caustics of the regions of $\Cds$ where the curves are dense. In Section~\ref{sec:max} we show how to apply this theory to find the supremum of the probabilities of transfer in any quantum walk.

\section{Factoring polynomials} \label{sec:polynomials}

This section assumes basic knowledge about field extensions of rationals. We recommend the main reference \cite{cox2012galois} for an introductory treatment, for example as a reference to Theorem \ref{thm:prim}. Following, we review algorithmic aspects of the theory, based mainly on  \cite{Trager1976,Landau1985}.

If $\alpha$ is a real number, root of a polynomial with rational coefficients, then $\alpha$ is an algebraic number. If the polynomial has integer coefficients and is monic, then $\alpha$ is an algebraic integer. Its minimal polynomial is the polynomial of smallest degree having $\alpha$ as a root. The field extension $\Qds[\alpha]$ contains $\Qds$ and all rational linear combinations of powers of $\alpha$ all the way up to the degree of its minimal polynomial minus one. Given a monic polynomial $p(x)$ with coefficients in $\Zds$, its splitting field is the smallest field extension of $\Qds$ over which $p(x)$ factors completely.

\begin{theorem}[Primitive Element Theorem] \label{thm:prim}
	If $p(x)$ is a monic polynomial with integer coefficients, then there exists $\alpha \in \Cds$ so that the splitting field of $p(x)$ is $\Qds[\alpha]$.
\end{theorem}

For us, $p(x)$ will be the minimal polynomial of $A(G)$, for some graph $G$. Its roots are the eigenvalues $\seq\theta0d$ of $A$, which play a major role in the behaviour of $\exp(\ii t A)$. Our goal in this section is to describe how to find polynomials $p_0,...,p_d$, with rational coefficients, so that if $\alpha^*$ is a primitive element for the splitting field of $p(x)$, we have $p_r(\alpha^*) = \theta_r$, for each $r$. This is equivalent to the task of completely factoring $p(x)$ over $\Qds[\alpha^*]$.

In \cite{Landau1985}, an algorithm of relative efficiency to completely factor a polynomial over its splitting field is presented. It consists of a clever application of the famous $L^3$ algorithm \cite{LLL} in conjunction with techniques to compute and factor the norm of polynomials over extension fields of the rationals. As a consequence of their work, we state a theorem for later reference that summarizes what we need in this paper. It is essentially \cite[Theorem 2.1]{Landau1985}.

\begin{theorem} \label{thm:computingpols}
	Given $A(G)$ with eigenvalues $\theta_0,...,\theta_d$, it is possible to recover polynomials $p_0,...,p_d$, with rational coefficients so that, for some primitive element $\alpha$ of the splitting field of the characteristic polynomial $\phi(x)$ of $A$, we have $p_r(\alpha) = \theta_r$. Moreover, the complexity of this procedure is polynomial on the degree of the splitting field of $\phi(x)$ over $\Qds$ and the logarithm of its largest coefficient. \qed
\end{theorem}

As a first application of this result to quantum walks, we show below that the entries of the average mixing matrix can be computed with exact precision. The average mixing matrix has been introduced in \cite{TamonAdamczakUniformMixingCycles}, and extensively studied in \cite{GodsilAverageMixing,CoutinhoGodsilGuoZhanAMM}. It is the matrix that gives the average of the probabilities of the quantum walk, that is
\[
	\widehat{M} = \lim_{T \to \infty} \frac{1}{T} \int_0^T \exp(\ii t A) \circ \exp(-\ii t A)\ \textrm{d}t,
\]
where $\circ$ stands for what is known as the entrywise, Hadamard or Schur product of matrices. Recalling that $A = \sum_{r = 0}^d \theta_r E_r$ is the spectral decomposition of $A$, it is easy to derive that
\[
	\widehat{M} = \sum_{r = 0}^d E_r \circ E_r,
\]
and in \cite[Lemma 2.4]{GodsilAverageMixing}, it is proved that the average mixing matrix is rational, with Lemma 3.1 therein giving an upper bound to the denominator of the entries. Here we show how to explicitly compute these rational numbers in exact arithmetics.

It is possible to express the entries of the idempotents using only the eigenvalues of the graph and of its vertex-deleted subgraphs (see \cite[Section 2.3]{coutinho2021quantum} and \cite[Chapter 4]{GodsilAlgebraicCombinatorics}). In what follows, we are denoting arbitrary vertices in the graph $G$ by $a$ and $b$, and by $\phi_H(x)$ the characteristic polynomial of a graph $H$. Then,

\begin{align}\label{eq:1}
	\bra{a} E_r \ket{a} = \frac{(x - \theta_r) \phi_{G \backslash a}(x) }{\phi_G(x)}\bigg|_{x = \theta_r} ,
\end{align}
and
\begin{align}\label{eq:2}
	\bra{a} E_r \ket{b} = \frac{(x - \theta_r) \sqrt{\phi_{G \backslash a}(x)\phi_{G \backslash b}(x) - \phi_{G}(x)\phi_{G \backslash ab}(x)}}{\phi_G(x)}\bigg|_{x = \theta_r}.
\end{align}

It is also known that the square root is indeed a polynomial, which we denote by $\phi_{ab}(x)$ when $a\neq b$. For ease of notation in the argument, let $\phi_{aa}(x) = \phi_{G \backslash a}(x)$.

\begin{theorem}\label{thm:avgmatrix}
	The integers in the numerators and denominators of the rational entries of $\widehat{M}$ are computable in exact arithmetics.
\end{theorem}
\begin{proof}
	Let $\alpha$ be a primitive element to the splitting field of $\phi(x)$, and consider polynomials $p_r(x)$ with $p_r(\alpha) = \theta_r$. Let $a$ and $b$ be vertices, possibly equal. First, divide numerator and denominator of the ratio $(x - p_r(\alpha))\phi_{ab}(x)/ \phi(x)$ by their gcd, computed over $\Qds[\alpha]$. Then, make $x = p_r(\alpha)$. Thus, $\bra{a} E_r \ket{b}$ is written as a polynomial with rational coefficients on the variable $\alpha$, and because 
	\[
		\bra{a}\widehat{M}\ket{b} = \sum_{r = 0}^d \bra{a} E_r \ket{b}^2,
	\]
	it follows that entries of $\widehat{M}$ will be expressed as polynomials in $\alpha$ as well. However this matrix is rational, as we mentioned, and thus, upon computing these polynomials over $\Qds[\alpha]$, we will have recovered the rational numbers.
\end{proof}

\section{Deciding pretty good state transfer} \label{sec:kronecker}

Pretty good state transfer between vertices $a$ and $b$ in a graph occurs whenever, for all $\varepsilon > 0$, it is possible to find $t$ so that
\[
	\big|\bra{a} \exp(\ii t A) \ket{b} \big| > 1 -\varepsilon.
\]
If there is a $t$ for which $|\bra{a} \exp(\ii t A) \ket{b}| = 1$, then we say perfect state transfer occurs. Given a graph, an algorithm that decides whether or not it admits perfect state transfer was shown in \cite{CoutinhoGodsilPSTpolytime}. Prior to our work, no algorithm that decides whether or not pretty good state transfer occurs was known. The difference between the two phenomenon is not insignificant: there are infinitely many examples of graphs that admit pretty good state transfer but not perfect, and the capability to identify more examples, or rule out candidates, is key to the design of new communication protocols within a quantum information framework.

Recall that $A = \sum \theta_r E_r$ is the spectral decomposition of $A$. From \cite[Lemma 3]{Banchi2017}, we know that pretty good state transfer implies that, whenever $E_r \ket a \neq 0$, it must be that
\[E_r \ket a = \sigma_r E_r \ket b,\]
with $\sigma_r = \pm 1$. This conditions is named strong cospectrality between vertices $a$ and $b$. 

A characterization of pretty good state transfer was provided in \cite{Banchi2017} using Kronecker's theorem on Diophantine approximations, building upon previous works. However, this characterization does not provide an algorithm that decides the existence of pretty good state transfer. Our goal below is to provide this algorithm. For the next lemma, we follow \cite[Lemmas 2.5 and 2.8]{LippnerEisenbergKempton}.

\begin{lemma}\label{lem:factorizationphi}
	Let $A$ be the adjacency matrix of a graph $G$, and assume vertices $a$ and $b$ are strongly cospectral. Assume $\phi(x)$ is the characteristic polynomial of $A$. Then $\phi(x)$ factors over $\Zds[x]$ as
	\[\phi(x) = \phi_+(x)\ \phi_-(x)\ \phi_0(x),\]
	and, moreover,
	\begin{enumerate}[(a)]
		\item The roots of $\phi_+$ and $\phi_-$ are simple.
		\item For each $\lambda$ root of $\phi_+$, there is eigenvector $\ket v$ of $A$ with $\braket{a}{v} = \braket{b}{v} \neq 0$.
		\item For each $\lambda$ root of $\phi_-$, there is eigenvector $\ket v$ of $A$ with $\braket{a}{v} = - \braket{b}{v} \neq 0$.
		\item For each root $\lambda$ of $\phi_0$ of multiplicity $k$, there are $k$ linearly independent eigenvectors of $A$ which are $0$ at $a$ and $b$.
		\item $\phi_+$ and $\phi_-$ share no common root.
	\end{enumerate}
	Moreover, $\phi_+$ is the minimal polymomial of $A$ in the module generated by $\ket{a} + \ket{b}$, and $\phi_-$ is the minimal polynomial of $A$ in the module generated by $\ket{a} - \ket{b}$.
\end{lemma}

With this lemma, Kronecker's theorem (see for instance \cite[Chapter 3]{AlmostPeriodicFunctionsBook}) gives us the right tool to characterize pretty good state transfer (see \cite[Theorem 2]{Banchi2017} or \cite[Lemma 2.10]{LippnerEisenbergKempton}).

\begin{theorem}[Kronecker's theorem] \label{thm:kro}
Let $\theta_0,...,\theta_d$ and $\zeta_0,...,\zeta_d$ be arbitrary real numbers. All systems of inequalities
\[| \theta_r y - \zeta_r | < \varepsilon \pmod {2\pi} , \quad (r = 0,...,d),\]
obtained for all $\varepsilon > 0$ admit a solution $y \geq 0$ (depending on $\varepsilon$) if and only if whenever integers $\ell_0,...,\ell_d$ satisfy
\[\ell_0 \theta_0 + ... + \ell_d \theta_d = 0,\]
they also satisfy
\[\ell_0 \zeta_0 + ... + \ell_d \zeta_d \equiv 0 \pmod {2\pi}.\] 
\end{theorem}

\begin{corollary}[Pretty good state transfer characterization] \label{cor:pgst}
	Given a graph $G$ with adjacency matrix $A$, characteristic polynomial $\phi$, and vertices $a$ and $b$, then there is pretty good state transfer between $a$ and $b$ if and only if both conditions below hold.
	\begin{enumerate}[(1)]
		\item Vertices $a$ and $b$ are strongly cospectral (consider the factorization of $\phi$ as in Lemma \ref{lem:factorizationphi}).
		\item Let $\{\lambda_i\}$ be the roots of $\phi_+$ and $\{\mu_j\}$ be the roots of $\phi_-$. For all integers $\ell_i$ and $m_j$ satisfying
		\[\sum_i \ell_i \lambda_i + \sum_j m_j \mu_j = 0 \quad \text{and} \quad \sum_i \ell_i + \sum_j m_j = 0,\]
		it also holds that
		\[\sum_j m_j \quad \text{is even}.\]
	\end{enumerate}
\end{corollary}

It is also known that condition (1) above can be tested in time polynomial on the number of vertices of $G$ (see for instance \cite{CoutinhoGodsilPSTpolytime}, but also as consequence of Lemmas 2.5 and 2.8 in \cite{LippnerEisenbergKempton}). We are ready to present our pretty good state transfer algorithmic characterization.

\begin{theorem}\label{thm:pgstdecision}
	There exists an algorithm that tests whether condition (2) of Corollary \ref{cor:pgst} holds or not. It works in time polynomial on the number of vertices and on the degree of the splitting field of $\phi(x)$.
\end{theorem}
\begin{proof}
	From Theorem \ref{thm:computingpols}, we know that we have polynomials $p_r(x)$ so that $p_r(\alpha) = \theta_r$, for all eigenvalues $\theta_r$ of the graph, and $\alpha$ a primitive element to the splitting field of $\phi(x)$. We may of course assume the degrees of the polynomials are smaller than the degree of the minimal polynomial of $\alpha$.
	
	Upon factoring $\phi$ as in Lemma \ref{lem:factorizationphi}, and because $\phi_+$ and $\phi_-$ share no common factor, we can identify which of the polynomials $p_r$ correspond to roots of $\phi_+$, and which correspond to roots of $\phi_-$.
	
	A linear combination of these polynomials evaluated at $\alpha$ will be equal to $0$ if and only if the linear combination of the polynomials is the zero polynomial, because their degrees are smaller than the degree of the minimal polynomial of $\alpha$. Thus, the equations from condition (2) in Corollary \ref{cor:pgst},
	\[\sum_i \ell_i \lambda_i + \sum_j m_j \mu_j = 0 \quad \text{and} \quad \sum_i \ell_i + \sum_j m_j = 0,\]
	give rise to a homogeneous linear system with rational coefficient matrix, which can be scaled to a linear system with integer coefficients, each equation having the gcd of its coefficients equal to 1. Assume $M x = 0$ expresses this system for some matrix $M$ and vector of variables $x$, and let $U$ and $V$ be invertible integral matrices, with integral inverses, of convenient size so that $UMV$ is the Smith normal form of matrix $M$. As a consequence, $(UMV) y = 0$ is trivial to solve, having the first $\ell$  variables in $y$ equal to $0$, with $\ell = \rk M$, and the remaining variables of free choice. Now $V y \in \Zds^m$ if and only if $y \in \Zds^m$, because $V^{-1}$ is also an integral matrix, thus $V y$ is a complete integral parametrization of the solution set of $Mx = 0$, with $x \in \Zds^m$. With it, we can write $\sum_j m_j$ as a sum of free integral variables with integer coefficients, and it is easy to verify that this sum is always even if and only if all coefficients are even.
	
	The complexity of finding polynomials $p_r$ is given by Theorem \ref{thm:computingpols}, and the complexity of computing the Smith normal form is well known to be polynomial on the order of the matrix with bounded coefficients (see for instance \cite{StorjohannThesis}).
\end{proof}

\section{Geometry of a quantum walk} \label{sec:shape}

In the previous section we showed how to use Theorem~\ref{thm:computingpols} to decide whether pretty good state transfer occurs. This is equivalent to asking whether the curve in the complex plane given by an off-diagonal entry of $\exp(\ii t A)$ approaches the unit circle. 

In this section we describe how Theorem \ref{thm:computingpols} provides the necessary tools to understand the geometry of these curves in general.

Recall the notation $U(t) = \exp(\ii t A)$, and that $A = \sum_{r = 0}^d \theta_rE_r$ denotes the spectral decomposition of $A= A(G)$. We begin with an example. Let $G = C_5$, the cycle graph on five vertices. The curves in the complex plane given by $U(t)_{1,1}$ and $U(t)_{1,2}$, with $t \in [0,100\pi]$, are, respectively shown in Figure \ref{figure1}.

\begin{figure}[h]
\includegraphics[scale=0.5]{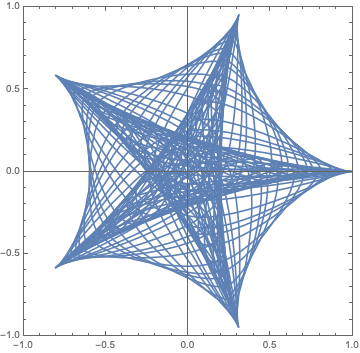}\includegraphics[scale=0.5]{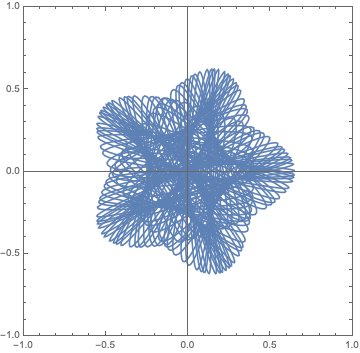} \caption{Entries $U(t)_{1,1}$ and $U(t)_{1,2}$ for $C_5$, $t \in [0,100\pi]$. Note that $U(t)_{1,1}$ starts at the point $(1,0)$ and therefore, as it is an almost periodic function, it approaches the unit circle arbitrarily often. On the other hand, pretty good state transfer does not happen between $1$ and $2$, as the closure of $U(t)_{1,2}$ is contained in circle of diameter smaller than $1$.} \label{figure1} 
\end{figure}

We are mostly interested in understanding some of the geometric features of these plots. 

Recall that
\[
	\bra{a} U(t) \ket{b} = \sum_{r = 0}^d \bra{a} E_r \ket{b} \e^{\ii t \theta_r}.
\]
This is an almost periodic function of $t$ (see \cite{AlmostPeriodicFunctionsBook}). Its behaviour depends crucially on rational dependences between the frequencies $\theta_r$ or, more precisely, between those $\theta_r$ for which $\bra{a} E_r \ket{b}\neq 0$. There are two extreme cases: On one extreme, all $\theta_r$ are integer up to a common factor, and the function is strictly periodic. At the other extreme the $\theta_r$ are rationally independent. Note that even though all eigenvalues of the graph sum to $0$, it could be that those for which $\bra{a} E_r \ket{b}\neq 0$ are indeed rationally independent, for some $a$ and $b$ vertices of the graph.

By Kronecker's Theorem, all these cases are covered by specifying the subgroup of the torus on which the tuple $(\e^{\ii t\theta_0},...,\e^{\ii t\theta_d})$ lies. This gives us two complementary ways of looking the value distribution of this function, which are visualized by Fig.~\ref{figure1} and Fig.~\ref{figure2}, respectively. The first figure just follows the values $\bra{a} U(t) \ket{b}$ as traced out with varying $t$ in the complex plane. By almost periodicity, the proportion of time spent on average in any region of the plane is a probability measure, called the sojourn measure. The second perspective looks at this as the image of a similarly defined distribution in the torus group: In the group the sojourn measure is clearly translation invariant, hence equal to the Haar measure. Therefore if we plot the image of a regular grid in the appropriate subgroup, we get a more direct representation of the sojourn measure (cp.\ Fig.~\ref{figure2}). In fact, when the $\theta_r$ are very nearly dependent, the direct orbit picture (run for a finite time) will be indistinguishable from the rational case, and hence an inaccurate representation of the infinite time limit. 

\begin{figure}[h]
\includegraphics[scale=0.5]{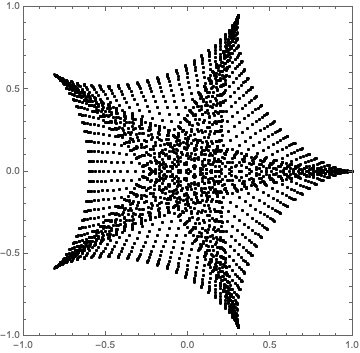}\includegraphics[scale=0.5]{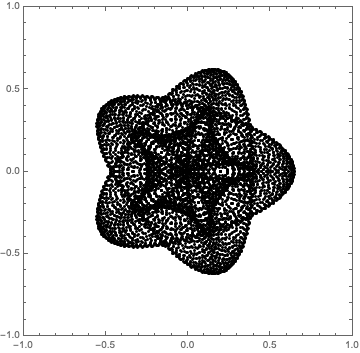} \caption{Images of the uniform grid of the $2$-torus onto the complex plane, under the map 
$(z_1,z_2) \mapsto \sum_{r = 0}^2 \bra{a} E_r \ket{b} \e^{\ii f^r(\ov{z})}$, where $E_r$ are the idempotents of $A(C_5)$, $f^0(z_1,z_2) = -2(z_1+z_2)$, $f^1(z_1,z_2) = z_1$, $f^2(z_1,z_2) = z_2$, and we have respectively $a = b$ and $a$ and $b$ distinct neighbours.} \label{figure2}
\end{figure}

In the following subsection we present an explicit description of how to compute these functions $f^x$ described in Figure~\ref{figure2} and therefore how to obtain this measure.

\subsection{The Haar measure of a quantum walk}

As before, assume $\phi(x)$, the characteristic polynomial of $A$, has been completely factored over its splitting field, and rational polynomials $p_r(x)$ as in Theorem \ref{thm:computingpols} have been computed. We will assume that $\bra{a} E_r \ket{b} \neq 0$ for all $r$, but if this is not the case, the set of indices for which this holds can be exactly determined using Equation \ref{eq:2}, and we can restrict the treatment that follows to those indices. 

Let $P$ be a matrix whose columns correspond to the polynomials $p_r$, and let $m \in \Zds$ so that $mP$ is integral. Assume the rank of $P$ is equal to the integer $k$. Upon computing the Smith normal form $m^{-1}S = UPV$, it follows that the non-zero columns of $m^{-1}U^{-1}S$ form a basis for the column space of $P$ that generates each column of $P$ as an integral linear combination, given by the columns of $V^{-1}$. Thus, the non-zero columns of $m^{-1}U^{-1}S$ correspond to algebraic numbers $w_1,...,w_k$, which are rationally independent, and so that there are integral linear combinations of these giving each $\theta_r$, that we denote by
\begin{equation}
	\theta_r = \sum_{\ell = 1}^k f^r_\ell w_\ell = f^r(\ov{w}).
\label{eq:th_r}
\end{equation}
Theorem \ref{thm:computingpols} therefore guarantees that the integer coefficients $f_\ell^r$ can be computed somewhat efficiently (they appear exactly in the first rows of $V^{-1}$). Note in particular that the linear space of rational combinations of the eigenvalues that are equal to $0$ is exactly generated by the last $(n-k)$ columns of $V$.

It follows from Kronecker's Theorem (Theorem \ref{thm:kro}) that $\{ t \ov{w} : t \in \Rds_+\}$ is dense on the $k$-dimensional torus $\Tds^k = \Rds^{k} / 2\pi\Zds^{k}$, and as $t_j$, for $j \to \infty$, are then chosen so that $t_j\ov{w}$ approximates the point $\ov{z} \in \Tds^k$, it must be that $t_j \theta_r$, for $r = 0 ,...,d$, will approximate the point
\[
	\zeta_r = f^r(\ov{z}) \pmod{2\pi}.
\]
The consequence is that there is a subspace of dimension $k$ in the $(d+1)$-torus where the curve
\[t \mapsto t (\theta_0,...,\theta_d) \pmod{2\pi}\] 
is dense. Moreover, a theorem due to Weyl, which extends Kronecker's Theorem, asserts that the curve $t \ov{w}$ covers the $k$-torus $\Tds^k$ uniformly, thus, if $\chi_S$ is the characteristic function of a Jordan measurable subset $S$ of $\Tds^k$, we have
\[
	\lim_{T \to \infty} \frac{1}{T} \int_0^T \chi_S(t \ov{w}) \ \textrm{d}t = \int_S \chi_S \ \textrm{d}V,
\]
where the right hand side is clearly the volume of $S$. For more details, see \cite[Chapter 1]{beck2017strong}. For this reason, in order to understand the region covered by the curve $\bra{a} U(t) \ket{b}$ in $\Cds^2$ and also how densely this happens in each subregion, we can introduce coordinate variables $z_1,...,z_k$ to the $k$-torus and consider the map given by
\begin{equation}
	F : (z_1,...,z_k) = \ov{z} \mapsto \sum_{r = 0}^d \bra{a} E_r \ket{b} \e^{\ii f^r(\ov{z})}. \label{eq:map}
\end{equation}
To exemplify, let us look again to the cycle $C_5$. Let $\gamma = \ee^{\ii( 2\pi/5)}$. The distinct eigenvalues of $C_5$ are $2$, $\gamma^2 + \gamma^3$ and $-1-\gamma^2-\gamma^3$, thus, with $w_1 = \gamma^2 + \gamma^3$ and $w_2 = -1 - w_1$, which are rationally independent, we have the eigenvalues $-2(w_1+w_2),w_1,w_2$. Upon examining the map described just above, the image of the uniform grid of the $2$-torus is quite similar to what we saw in Figure \ref{figure1}, as can be seen in Figure \ref{figure2}.

One observation about these pictures is that they will always be symmetric about the real axis, even for $t \in \Rds_+$ only. This is an immediate consequence of Theorem \ref{thm:kro}, because if the condition in its statement holds for $\{\zeta_0,...,\zeta_d\}$, then it also holds for $\{-\zeta_0,...,-\zeta_d\}$.

\begin{corollary}
	Let $G$ be a graph, and $A=A(G)$. Then the closure of the curve in the complex plane described by any entry of $\exp(\ii t A)$ with $t \geq 0$ and as $t \to \infty$ is invariant under complex conjugation. \qed
\end{corollary}

\subsection{The uninteresting cases} \label{sec:uninteresting}

Given the spectral decomposition of the graph, $A = \sum_{r = 0}^d \theta_rE_r$, it could occur that $\langle a | E_r | b \rangle \neq 0$ implies $\theta_r \in \Zds$. In this case, the image of $U(t)_{a,b}$ in the complex plane will coincide with the image of the $1$-torus under an injective map, and therefore it will result in a closed curve, with period $2\pi$. Some of these curves might be interesting on their own, but the questions are certainly going to be more simply addressed. An interesting exercise is to plot the curves obtained from the well-known Petersen graph.

A second uninteresting case is that of a bipartite graph. In this case, the adjacency matrix can be written as
\[
	A =\pmat{0&B\\ B^T&0},
\]
where $B$ is a $01$ matrix of appropriate size. Then
\[
	A^{2k} = \pmat{(BB^T)^k&0\\ 0&(B^TB)^k},\qquad A^{2k+1} = \pmat{0&(BB^T)^kB\\ (B^TB)^kB^T&0}
\]
and
\[
	U(t) = \exp(\ii tA) 
		= \pmat{\cos(t\sqrt{BB^T})& \ii\sin(t\sqrt{BB^T})B\\ 
			\ii\sin(t\sqrt{B^TB})B^T& \cos(t\sqrt{B^TB})}.
\]
As a consequence, for all times $t$ and any two vertices $a$ and $b$, either $U(t)_{a,b}$ is always real or always purely imaginary. The plots of these curves on the complex plane will hence be entirely contained in the coordinate axis, and not much will be seen. If the graph is not bipartite, then the largest eigenvalue has strictly larger absolute value than any other eigenvalue, and therefore any entry of $U(t)$ will attain values which are neither real nor purely imaginary.

Fortunately, most graphs are neither bipartite nor have integer spectrum, so one should expect that the typical case is interesting.

\subsection{Odd prime cycles} \label{sec:oddprime}

Figure \ref{figure1} displays a rotational symmetry. Despite a first guess, this symmetry is not related to a graph automorphism of the cycle, but rather to the fact that adding a certain constant angle to the free independent variables described in Equation \eqref{eq:th_r} results in adding the same constant to all eigenvalues. This phenomenon is common to all odd prime cycles.

\begin{theorem}
	Let $G = C_p$, with $p$ an odd prime, and $A=A(G)$. Then any entry of $\exp(\ii t A)$, as $t \to \infty$, is dense in a region $R$ of $\Cds^2$ that admits a $p$-fold rotational symmetry of the plane.  
\end{theorem}
\begin{proof}
	Let $\omega = \exp(2 \pi \ii/p)$. Distinct eigenvalues of $A(C_p)$ are $\theta_r = \omega^r + \omega^{p-r}$, for $r \in \{0,1,\cdots,(p-1)/2\}$. All have multiplicity equal to $2$, except for $\theta_0$ that is simple. First note that $\Qds[\omega]$ contains all eigenvalues, and that the minimal polynomial of $\omega$ is the $p$th cyclotomic polynomial $\phi_p(x) = \sum_{i = 0}^{p-1} x^i$. From the eigenvalue expressions, it is immediate that $\ov{\theta} = \{\theta_1,\cdots,\theta_{(p-1)/2}\}$ form a basis for the extension, with $\theta_0$ equal to $-2$ times their sum.
	
	As shown in Section \ref{sec:shape}, the region $R$ is determined by the map from the torus to the complex plane shown in Equation \eqref{eq:map}. If we add $2\pi/p$ to each torus coordinate, then each $f^r(\ov{z})$ gets added by the same $2\pi/p$. This is readily seen for $r > 0$ as $f^r$ is a coordinate projection, and it is true for $f^0$, because

\begin{align*}
    f^0(\ov{z} + (2\pi/p) \1) &\equiv_{2\pi} -2 \sum_{i = 1}^{(p-1)/2} \left(z_i + \frac{2\pi}{p}\right)\\
     & \equiv_{2\pi} -\frac{2 (p-1)\pi}{p} -2 \sum_{i = 1}^{(p-1)/2} z_i\\
     & \equiv_{2\pi} \frac{2\pi}{p} -2 \sum_{i = 1}^{(p-1)/2} z_i \\
     & \equiv_{2\pi} \frac{2\pi}{p} + f^0(\ov{z}).
\end{align*}
\end{proof}

Note that the result above does not hold for other odd cycles, as more rational dependences occur amongst their eigenvalues. For example, the entries of $U(t)$ for the cycle $C_9$ only exhibit $3$-fold rotational symmetry. For a graph in general, we note that an $n$-rotational symmetry arises whenever, for all $r$, we have
\begin{align}
	\sum_{\ell = 1}^k f_\ell^r \equiv 1 \pmod n, \label{eq:rotation}
\end{align}
with $f_\ell^r$ the coefficients described in Equation \eqref{eq:th_r}.

Another distinguished geometric feature of the pictures in Figure \ref{figure2} is related to the singularities of the map from the torus to the plane. The distinguished borders appearing in Figure \ref{figure2} and resembling star like contours are the caustics of this map. They are the images under $F$ (described in Equation \eqref{eq:map}) of the curves in the torus which are the points where the Jacobian matrix of $F$ does not have full rank.

For the cycle $C_5$, we first analyze the diagonal entries. Here, Equation \eqref{eq:map} reduces to
\[
	F(z_1,z_2) = \frac{1}{5} \e^{-2 \ii (z_1 + z_2)} + \frac{2}{5} \e^{\ii z_1} +   \frac{2}{5} \e^{\ii z_2},
\]
thus $\partial F / \partial z_1$ is (real) parallel to $\partial F / \partial z_2$ precisely when $z_1 = z_2$ or $2z_1 = -3z_2$ (and equivalently $-3z_1 = 2z_2$). These solutions describe the hypocycloids in Figure \ref{figure3}. 

\begin{figure}[h]
\begin{center}
\includegraphics[scale=0.5]{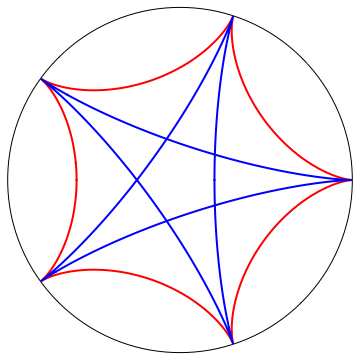} \caption{Cycloids $\e^{-\ii 4t} + 4 \e^{\ii t}$ (red, outer curve) and $3 \e^{2 \ii t} + 2 \e^{-3\ii t}$ (blue, inner curve), up to scalling.} \label{figure3}
\end{center}
\end{figure}

In general, for a diagonal entry of $\exp(\ii t A(C_p))$, $p$ an odd prime,  the following hypocycloids are obtained:
\[
	\frac{1}{p} \big( k \e^{-\ii (p-k) t} + (p-k) \e^{\ii k t} \big),
\]
for $k$ between $1$ and $(p-1)/2$.

For the off-diagonal entries, the solutions cannot be so easily expressed. For instance, for the cycle $C_5$ and an entry corresponding to neighbours, Equation \eqref{eq:map} reduces to
\[
	F(z_1,z_2) = \frac{1}{5} \left( \e^{-2 \ii (z_1 + z_2)} + \frac{-1+\sqrt{5}}{2} \e^{\ii z_1} + \frac{-1-\sqrt{5}}{2} \e^{\ii z_2}\right),
\]
and an expression for the exact solution to $\partial F/\partial z_1= \gamma \cdot \partial F/\partial z_2$ with $\gamma \in \Rds$ is not available, although it is always possible to obtain a numerical approximation of the curve. The problem however becomes significantly less tractable for larger cycles or other graphs.

\subsection{More graphs}\label{sec:another}

Let $G$ be the graph $K_4$ minus one of its edges\footnote{The graphs $K_n$, where $n$ vertices are all connected to each other, have integer eigenvalues, and therefore the plots of the curves $\exp(\ii t A(K_n))_{a,b}$ will be uninteresting.}. Figure \ref{fig:4gallery} displays four distinct entries of $\exp(\ii t A(G))$ for $t$ in $[0,100]$, with the insets at top left displaying which.

\begin{figure}[!h]
  \centering
   \includegraphics[width=5cm]{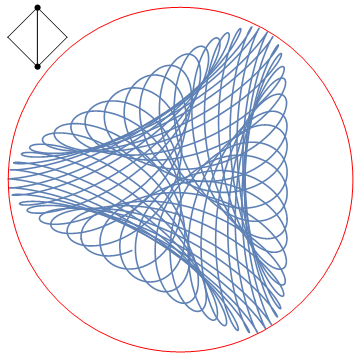}\\
    \includegraphics[width=5cm]{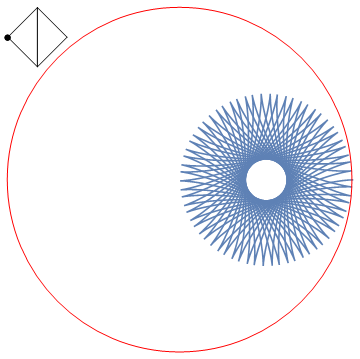}\ \includegraphics[width=5cm]{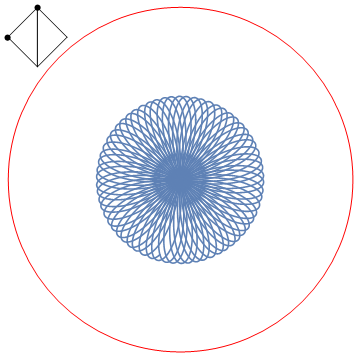}\\
  \caption{Sketches of various matrix element functions of the quantum walk on the graph $K_4$ minus one edge.}\label{fig:4gallery}
\end{figure}

The graph has characteristic polynomial $\phi(t)=t(t+1)(t^2 - t -4)$. The eigenvector for the eigenvalue $0$ is nonzero only at the two vertices of degree $2$, so the term corresponding to it only plays a role in the bottom left picture, leading to the translation of the center of mass. Having $0$ as an eigenvalue and the entry in its eigenprojector nonzero is the only circumstance these picture will not be centered at the origin.
   
The eigenvalues simultaneously in the support of the vertices of degree $2$ and $3$ are precisely the roots of $(t^2 - t -4)$. Thus they are rationally independent, and this is quite special. It implies Equation \eqref{eq:rotation} holds for all $n$, thus the (closure of the) picture in the bottom right is fully rotationally symmetric.

Finally, the picture in the top is an example of pretty good state transfer, that we discussed in Theorem \ref{thm:pgstdecision}. The three-fold symmetry arises from the fact that the eigenvalues in the support of the two vertices are the roots of $(1 + t)(t^2 - t -4)$, and Equation \eqref{eq:rotation} and the fact that two of them sum to minus the third implies the rotational symmetry.

This shows that a crucial role is played by the rational dependences of the eigenvalues of the graph. In other terms, let us call $\Gamma$ the closure of the subgroup generated by all the tuples $(t \theta_k)$, $k=0,...,d$ in the $(d+1)$-torus. Clearly, the map sending this to $\exp(\ii t A)_{a,b}$ can be thought of as a continuous function on $\Gamma$, and since the average of a function of time is invariant under time translations, the averaging process corresponds to a translation invariant average on $\Gamma$, i.e., to the Haar measure. Now by duality of locally compact groups the group $\Gamma$ is uniquely characterized by the set of characters vanishing on it, i.e., by the integer tuples $(n_k)$ such that $\sum_k n_k \theta_k=0$. These are just the rational dependences of the eigenvalues.

At this point a natural question is whether these rational dependences are always trivially obtained by factoring the characteristic polynomial over the integers and considering only the sums of eigenvalues corresponding to each factor. This is false for bipartite graphs, as any eigenvalue comes coupled with its negative even if they belong to factors of large degree. We display a non-bipartite example in Figure~\ref{fig:nonbip}.

\begin{figure}[!h]
\begin{center}
\includegraphics[scale=1]{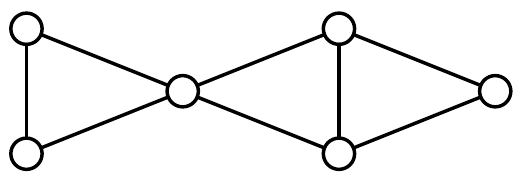} \caption{The characteristic polynomial admits a factorization over the integers as $(t - 1)^2(t^4 - 2 t^3 - 5 t^2 + 6 t + 4)$, and $\frac{1 \pm \sqrt{13 \pm 4\sqrt{5}}}{2}$ are the roots of the degree 4 term. They are paired into eigenvalues that sum to $1$, so there are dependences which do not require summing all roots of every factor.} \label{fig:nonbip}
\end{center}
\end{figure}

\section{The supremum of probabilities and even moments} \label{sec:max}

Given a graph $G$ and two vertices $a$ and $b$, what is the supremum of the probabilities of transfer of state from $a$ to $b$ during a quantum walk? Equivalently, what is the radius of the smallest unit disk in the complex plane that contains $\{\exp(\ii t A)_{a,b} : t \in \Rds_+\}$ ?

It is a well known fact that if $M$ is the maximum value attained by $|\exp(\ii t A)_{a,b}|$ with $t \in [r_1,r_2]$, then
\[
	\frac{1}{r_2-r_1} \cdot \lim_{m \to \infty} \left(\int_{r_1}^{r_2} |\exp(\ii t A)_{a,b}|^m \textrm{d}t \right)^{1/m} = M,
\]
Thus making $[r_1,r_2] \to [0,\infty)$ and choosing large enough values of $m$, one obtains progressively good approximations for $M$. Unfortunately however the computation of $|\exp(\ii t A)_{a,b}|^k$ is not efficient, and already for small graphs this procedure will not lead to satisfactory results.

The theory we presented in Sections \ref{sec:kronecker} and \ref{sec:shape} allow for an alternative and more effective approach (as long as we replace maximum for supremum).

\begin{proposition}
	Assume coefficients $f^r$ are as given in Equation \eqref{eq:th_r}, and $F$ as in Equation \eqref{eq:map}. The supremum of $|\exp(\ii t A)_{a,b}|$ for $t \in \Rds^+$ is
	\[
		\frac{1}{(2\pi)^k}\cdot \lim_{m \to \infty} \left(\int_{\mathds{T}^k} |F(\ov{z})|^m \textrm{d}\ov{z} \right)^{1/m}
	\] \qed
\end{proposition}

The benefit of this approach is that for $m$ even, $|F(\ov{z})|^m$ is a sum of exponentials, and for each term the exponent is an integer combination of independent torus variables. Hence a term survives the integral only when the coefficient of each variable in its exponent is equal to $0$. Thus, from
\begin{align*}
	|F(\ov{z})|^2 & = \left(\sum_{r = 0}^d \bra{a} E_r \ket{b} \e^{\ii f^r(\ov{z})}\right)\left(\sum_{r = 0}^d \bra{a} E_r \ket{b} \e^{-\ii f^r(\ov{z})} \right)\\& = \sum_{r,s = 0}^d \bra{a} E_r \ket{b}\bra{a} E_s \ket{b} \e^{\ii (f^r-f^s)(\ov{z})},
\end{align*}
it follows that the terms in $|F(\ov{z})|^{2\ell}$ whose exponents are equal to $0$ correspond to the solutions of $p_{00} + p_{01} + \cdots + p_{dd} = \ell$, with $p_{rs}$ nonnegative integers, and so that, for all $i$ from $1$ to $k$, we have
\[
	\sum_{r,s} p_{rs}(f^r_i - f^s_i) = 0.
\] 
This approach provides an exact method to approximate the supremum of probability of transfer in any quantum walk. Moreover, it allows us to draw a connection to the average mixing matrix discussed in Theorem \ref{thm:avgmatrix}. With $F$ coming from the definition in Equation \eqref{eq:map}, note that 
\[
	\frac{1}{(2\pi)^2} \int_{\mathds{T}^k} |F(\ov{z})|^2 \textrm{d}\ov{z}
\]
is precisely equal to the $ab$-entry of $\widehat{M}$, the matrix whose entries are the absolute second moments of the entries of $\exp(\ii t A)$. We show below that the even moments are all rational.

\begin{theorem}
	Assume coefficients $f^r$ are as given in Equation \eqref{eq:th_r}, and $F$ as in Equation \eqref{eq:map}, and that $\ell \geq 1$ an integer. Then
	\[
	\frac{1}{(2\pi)^{2\ell}} \int_{\mathds{T}^k} |F(\ov{z})|^{2\ell} \textrm{d}\ov{z}
	\]
	is rational.
\end{theorem}
\begin{proof}
	In the expansion of $|F(\ov{z})|^{2\ell}$, we consider the terms not multiplying the exponentials. We would like to show that their sum is invariant under all automorphisms of the splitting field of the characteristic polynomial of the graph. As discussed above, each of these terms corresponds to a solution of $p_{00} + p_{01} + \cdots + p_{dd} = \ell$, with $p_{rs}$ nonnegative integers, and so that, for all $i$ from $1$ to $k$, we have $\sum_{r,s} p_{rs}(f^r_i - f^s_i) = 0$. Assume the $w_i$ are as in Equation \ref{eq:th_r}. Then $\sum_{r,s} p_{rs}(f^r_i w_i - f^s_i w_i) = 0$, thus
	\begin{align*}
		0 & = \sum_i \sum_{r,s} p_{rs}(f^r_iw_i - f^s_iw_i) \\ & 
		= \sum_{r,s} p_{rs}\left(\sum_i f^r_i w_i - \sum_i f^s_i w_i\right) \\
		& = \sum_{r,s} p_{rs}(\theta_r - \theta_s).
	\end{align*}
	Let $\Psi$ be a field automorphism of the splitting field of the characteristic polynomial of the graph, and recall that $\Psi$ induces a permutation on the (indices of the) eigenvalues, whose inverse we shall denote by $\psi$. Then 
	\[
	0 = \Psi(0) = \Psi \left(\sum_{r,s} p_{rs}(\theta_r - \theta_s) \right) = \sum_{r,s} p_{\psi(r)\psi(s)}(\theta_r - \theta_s).
	\]
	As a consequence, the set of terms not multiplying exponentials is preserved by $\Psi$, and therefore so is their sum.
\end{proof}

\section{Conclusion} \label{sec:conclusion}

We are motivated by the task of understanding as much as possible about the quantum walk in a given arbitrary graph. 

A simple question such as what is the maximum probability of transfer between two vertices is still not completely addressed. We provided a method to decide whether the supremum is equal to $1$ or not. We were able to do this by exhibiting an algorithm that tests a special case of the well-known Kronecker's theorem on Diophantine approximations. In order to achieve this result, we used known techniques that were yet strange in the analysis of quantum walks in graphs.

Our investigation led to us to some deeper and more interesting questions about the geometry of the curves drawn by the entries of $\exp(\ii t A)$. We provided a thorough analysis of some geometric features, characterizing rotational symmetry and singularities. 

Following, we showed yet another improvement to the problem of finding the best possible probability of transfer, exhibiting a method that approximates the supremum up to desired precision even when not equal to $1$. This was achieved by applying some observations about how to compute even absolute moments, along with a result showing that these moments are rational. 

We speculate that further analysis of the moments of the entries of $\exp(\ii t A)$ could lead to interesting theory --- for instance, the torus is compact and the function $F$ defined in Equation \eqref{eq:map} is continuous, thus in principle the probability distribution of its image on the complex plane can be completely recovered from the $(a,b)$-moments obtained upon averaging $F(\ov{z})^a\ov{F(\ov{z})^b}$ over the torus.

Finally, a still unanswered question regarding pretty good state transfer is how to compute times that achieve high probability of transfer. The  interpretation of this phenomenon as a map from torus to the complex plane might lead to new advances regarding this question.

\subsection*{Acknowledgements}

Authors acknowledge the hospitality of the Banff International Research Station during the 2019 Quantum Walks and Information Tasks workshop, when seeds to this work were planted. P.F. Baptista acknowledges a CAPES master's scholarship. C. Godsil gratefully acknowledges the support of the Natural Sciences and Engineering Council of Canada (NSERC), Grant No.RGPIN-9439.

\bibliographystyle{plain}
\bibliography{qwalk.bib}

\end{document}